\documentclass[11pt,letterpaper]{article}
\usepackage{amsmath,amsthm,amssymb,xfrac,bm,xcolor,cases,hyperref,xspace,tikz}
\usepackage[margin=1in]{geometry}

\newcounter{thm}
\newtheorem{theorem}[thm]{Theorem}
\newtheorem{lemma}[thm]{Lemma}
\newtheorem*{lemma*}{Lemma}

\newtheorem{claim}[thm]{Claim}
\newtheorem{proposition}[thm]{Proposition}

\usepackage{complexity}

\theoremstyle{definition}
\newtheorem{definition}{Definition}

\newtheorem*{remark*}{Remark}

\newcommand{\cN}{\mathcal N}
\renewcommand{\R}{\mathbb R} %
\newcommand{\ip}[1]{\left\langle#1\right\rangle}

\newcommand{\trans}{^{\intercal}}
\newcommand{\bx}{\bm x}
\newcommand{\bv}{\bm v}

\newcommand{\bs}{\bm s}
\newcommand{\bt}{\bm t}
\newcommand{\by}{\bm y}
\newcommand{\be}{\bm e}
\newcommand{\bu}{\bm u}
\newcommand{\bg}{\bm g}

\newcommand{\eps}{\epsilon}
\newcommand{\supp}{\operatorname{supp}}

\newcommand{\spa}{\operatorname{sparse}}
\newcommand{\den}{\operatorname{dense}}

\title{Smoothed Complexity of 2-player Nash Equilibria}
\author{Shant Boodaghians\thanks{University of Illinois at Urbana-Champaign. Supported by NSF grant CCF-1750436} \and Joshua Brakensiek\thanks{Stanford University. Supported by an NSF Graduate Research Fellowship.} \and Samuel B. Hopkins\thanks{University of California, Berkeley. Supported by a Miller Postdoctoral Fellowship.} \and Aviad Rubinstein\thanks{Stanford University}}
\date{}

\begin{document}
\maketitle

\begin{abstract}
We prove that computing a Nash equilibrium of a two-player ($n \times n$) game with payoffs in $[-1,1]$ is \PPAD-hard (under randomized reductions) even in the smoothed analysis setting, smoothing with noise of \emph{constant} magnitude.
This gives a strong negative answer to conjectures of Spielman and Teng [ST06] and Cheng, Deng, and Teng [CDT09].

In contrast to prior work proving \PPAD-hardness after smoothing by noise of magnitude $1/\poly(n)$ [CDT09], our smoothed complexity result is \emph{not} proved via hardness of approximation for Nash equilibria.
This is by necessity, since Nash equilibria can be approximated to constant error in quasi-polynomial time [LMM03].
Our results therefore separate smoothed complexity and hardness of approximation for Nash equilibria in two-player games.

The key ingredient in our reduction is the use of a \emph{random zero-sum game} as a gadget to produce two-player games which remain hard even after smoothing.
Our analysis crucially shows that all Nash equilibria of random zero-sum games are far from pure (with high probability), and that this remains true even after smoothing.
\end{abstract}
\vfill

\begin{figure}
  \centering
	\def\r{1.5}
\def\w{0.03}
\def\highlight{blue!55!green!45!white}
	\begin{tikzpicture}
	   
	\begin{scope}[shift={(-6*\r,0)}]
	\node at (-1*\r, \r) {(a)};
    \draw[thick, fill=\highlight] (0,0) circle (\r);
    
    \begin{scope}[scale=\r]
    \edef\points{}
    \foreach \point [count=\i] in {
        (-0.0379132031724, -0.139664730507),
        (-0.528281338649, -0.124334980586),
        (-0.1565313227, -0.544895859619),
        (0.58679649386, -0.678806078824),
        (0.0650917380462, -0.59290965366),
        (-0.560797525839, -0.732710434918),
        (-0.154434040702, 0.383050907908),
        (0.762199229989, -0.32436671741),
        (0.325079323179, 0.225073809187),
        (0.087445083514, 0.354871059114),
        (-0.345868342326, 0.0238181286022),
        (-0.325233577705, 0.467097811294)}
    {
        \def\this{point-\i}
        \node[coordinate] (\this) at \point {} ;
        \fill[color=white] (\this) circle (\w) ;
        \xdef\points{(\this) \points}
    }
    \end{scope}
    
	\draw (-2*\w,-2*\w)--(2*\w, 2*\w);
	\draw (-2*\w,2*\w)--(2*\w, -2*\w);
	\draw[thin, <->] (-1*\r,0.4*\r) -- (0,0.4*\r) node [midway,above] {$\sigma$};
	\draw[thin] (-1*\r,0)--(-1*\r,0.5*\r) (0,0)--(0,0.5*\r);
	\node at (0,0) [below right] {$x$};
    \end{scope}

    \begin{scope}[shift={(-3*\r,0)}]
    \node at (-1*\r, \r) {(b)};
    \draw[thick] (0,0) circle (\r);
    
    \begin{scope}[scale=\r]
    \edef\points{}
    \foreach \point [count=\i] in {
        (-0.0379132031724, -0.139664730507),
        (-0.528281338649, -0.124334980586),
        (-0.1565313227, -0.544895859619),
        (0.58679649386, -0.678806078824),
        (0.0650917380462, -0.59290965366),
        (-0.560797525839, -0.732710434918),
        (-0.154434040702, 0.383050907908),
        (0.762199229989, -0.32436671741),
        (0.325079323179, 0.225073809187),
        (0.087445083514, 0.354871059114),
        (-0.345868342326, 0.0238181286022),
        (-0.325233577705, 0.467097811294)}
    {
        \def\this{point-\i}
        \node[coordinate] (\this) at \point {} ;
        \fill[color=\highlight] (\this) circle (\w) ;
        \xdef\points{(\this) \points}
    }
    \end{scope}
    
	\draw (-2*\w,-2*\w)--(2*\w, 2*\w);
	\draw (-2*\w,2*\w)--(2*\w, -2*\w);
	\draw[thin, <->] (-1*\r,0.4*\r) -- (0,0.4*\r) node [midway,above] {$\sigma$};
	\draw[thin] (-1*\r,0)--(-1*\r,0.5*\r) (0,0)--(0,0.5*\r);
	\node at (0,0) [below right] {$x$};
    \end{scope}

    \node at (-1*\r, \r) {(c)};
	\draw[thick, fill=white] (0,0) circle (\r);
	\draw (-2*\w,-2*\w)--(2*\w, 2*\w);
	\draw (-2*\w,2*\w)--(2*\w, -2*\w);
	\draw[thin] (-1*\r,0)--(-1*\r,0.5*\r) (0,0)--(0,0.5*\r);
	\draw[thin, <->] (-1*\r,0.4*\r) -- (0,0.4*\r) node [midway,above] {$\sigma$};
	\node at (0,0) [below right] {$x$};

	\end{tikzpicture}
   \caption{The $\sigma$-neighborhood of an intractable instance $x$. 
   Tractable instances are colored. 
   (1-a) {\bf smoothed-algorithmica}: almost all instances in the $\sigma$-neighborhood are tractable, (1-b) {\bf smoothed-complexity}: very few instances in the $\sigma$-neighborhood are tractable; (1-c) {\bf hardness of approximation}: it is intractable to find a solution to any instance in the $\sigma$-neighborhood.}
   \label{fig:three-cases}
\end{figure}
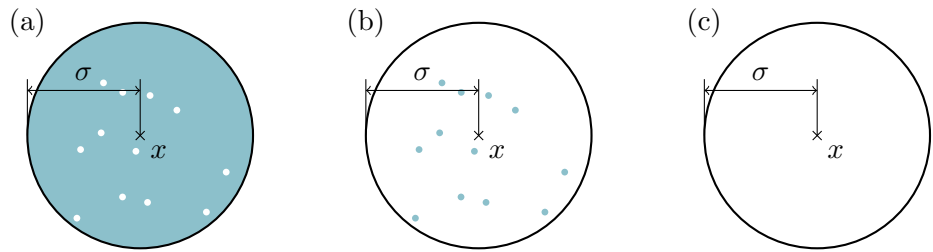

\setcounter{page}{0}
\thispagestyle{empty}
\newpage

\section{Introduction}\label{sec:intro}

Nash equilibrium is the central solution concept in game theory. 
Computational complexity results establishing the intractability of Nash equilibrium~\cite{savani2006hard,chen2009settling,daskalakis2009complexity} suggest that players that are even mildly computationally bounded may not be able to converge to a Nash equilibrium in the worst case.
However, the fragility of these intractable game constructions, together with the fact that random games are tractable~\cite{barany2007nash}, have led experts to conjecture that Nash equilibrium should have smoothed polynomial time algorithms (e.g.~\cite{spielman2006smoothed} Conjecture 15, and \cite{chen2009settling} Conjecture 2). 
If these conjectures were true, they could explain why players in realistic games can converge to equilibrium. In this paper, we prove that even with aggressive smoothing perturbations of constant magnitude, finding a Nash equilibrium continues to be \PPAD-complete (under randomized reductions).

\begin{definition}[\textbf{$X$-SMOOTHED-NASH}] \hfill

For a distribution $X$ on $\mathbb{R}$ and problem size $n$,
fix worst-case $n\times n$ matrices $W_A,W_B$ with entries in $[-1, 1]$,
and let $N_A, N_B$ be $n\times n$ matrices whose entries are drawn i.i.d.~from $X$.
$X$-SMOOTHED-NASH is the problem of computing,
with probability%
\footnote{Amplifying the success probability of smoothed algorithms is generally non-trivial (and sometimes impossible). We note that our hardness result continues to hold even for algorithms that are only required to succeed with probability $o(1)$.}
 at least $1 - \frac{1}{n}$, 
 a Nash equilibrium %
 of the game $(W_A+N_A, W_B + N_B)$.
\end{definition}

\begin{theorem}[Main Theorem]\label{thm:main} \hfill\\
There exists a universal constant $\eps > 0$, such that for any probability distribution $X$ on $[-\eps, \eps]$, $X$-SMOOTHED-NASH is \PPAD-hard under a randomized reduction.\footnote{Formally, we assume that there is a (single-dimensional) distribution  $X'$ such that $\Pr_{x \sim X, x' \sim X'}(|x - x'| > 1/\poly(n)) \leq 1/\poly(n)$ and $X'$ can be sampled by randomized polynomial-time algorithm. This holds for any natural smoothing distribution -- e.g. truncated Gaussian or uniform. For arbitrary $X$ such $X'$ can be sampled by a randomized algorithm which receives as input a $\poly(n)$-size approximation of the CDF of $X$. For $X$ where such advice is necessary, our arguments show that $X$-SMOOTHED-NASH is \PPAD-hard under randomized reductions with polynomial-length advice.}
\end{theorem}

\subsection{Complexity context: smoothed analysis vs hardness of approximation}
In their 2006 survey on smoothed analysis, Spielman and Teng posed the challenge (\cite{spielman2006smoothed}, Open Question 11) of exploring the connections between smoothed complexity and hardness of approximation. Concretely, they considered the example of two-player Nash equilibrium subject to $\sigma$-bounded perturbations: Given a hard game $A,B \in [-1,1]^{n \times n}$, perturbing each entry independently gives rise to a new instance $\hat{A},\hat{B} \in [-1-\sigma,1+\sigma]^{n \times n}$; any Nash equilibrium of $\hat{A},\hat{B}$ is an $O(\sigma)$-approximate-Nash equilibrium of the original game $A,B$.
 Hence, solving Nash equilibrium in the smoothed model is at least as hard as approximating Nash~\cite[Proposition 9.12]{spielman2006smoothed}.

More generally, any hard instance $x$ of any computational problem%
\footnote{Naturally, the correspondence between approximation algorithms and smoothed analysis requires matching the respective notions of approximation and smoothing perturbations.}
(e.g.~$x= (A,B)$ in the case of Nash) can be in one of three states (as illustrated in Figure~\ref{fig:three-cases}):
\begin{description}
\item[Smoothed-algorithmica\footnotemark:]
\footnotetext{The name is ``smoothed-algorithmica'' is a cultural reference to Impagliazzo's five worlds~\cite{impagliazzo1995personal} and should not be interpreted to imply a technical connection. 
Formally speaking, the existence of smoothed-algorithmica yet intractable instances is in fact consistent with any of Impagliazzo's worlds {\em except} Algorithmica.} 
Most instances in $x$'s neighborhood can be solved efficiently.
\item[Smoothed-complexity:] A small fraction of $x$'s neighborhood can be solved efficiently.
\item[Hardness-of-approximation:] Finding a solution for any instance in $x$'s neighborhood is intractable.\footnote{Intuitively we would like to say that every instance in $x$'s neighborhood is intractable. Formally, however, this may be inaccurate. In fact in the case of Nash equilibrium it is provably false! Given game $(A,B)$, consider game $(A',B)$, where $A' - A$ is a matrix whose entries are identically equal to some small $\lambda$ which encodes a Nash equilibrium for $(A,B)$ (and hence also for $(A',B)$).}
\end{description}

Of course, as in Spielman and Teng's proposition, hardness-of-approximation immediately rules out efficient smoothed algorithms. 
But most interesting open problems in smoothed analysis admit approximation algorithms; this limits the applicability of using hardness-of-approximation to prove new smoothed-complexity results.

In contrast to the thriving literature on hardness of approximation and smoothed algorithms, smoothed complexity results are rare. In this paper, we make a small step toward establishing a theory of smoothed complexity, in the context of Spielman and Teng's original example: two-player Nash equilibrium subject to bounded perturbations%
\footnote{Speilman and Teng discuss perturbing each entry by a uniform-$[-\sigma,\sigma]$ noise, but our result holds for any bounded i.i.d.~perturbations}.
While settling an open problem in equilibrium computation, we believe that our result is just the tip of the iceberg of the theory of smoothed complexity.

\subsection{Historical context}

In 1928 Von Neumann~\cite{neumann1928theorie} proved that every (finite, perfect information) zero-sum game has an equilibrium; this result was extended to general games by Nash in 1951~\cite{nash1951non}. In 1947 Dantzig~\cite{dantzig1998linear} designed the simplex algorithm for solving linear programs (and thus also zero-sum games); in 1964 Lemke and Howson~\cite{lemke1964equilibrium} gave a simplex-like algorithm for general games.
Both are known to take exponential time in the worst case~\cite{klee1972good,savani2006hard}, but are observed to perform much better in practice (e.g.~\cite{shamir1987efficiency,avis2009enumeration}).

 For linear programming, Khachiyan~\cite{khachiyan1979polynomial} gave the first polynomial time algorithm in 1979, and Spielman and Teng proved in 2004~\cite{spielman2004smoothed} that the simplex algorithm has smoothed polynomial complexity. It was natural to hope (and in fact quite widely believed, e.g.~\cite{daskalakis2005complexity} and~\cite[Conjecture 9.51]{spielman2006smoothed} respectively) that the last two results would again be extended to general games. Surprisingly, this was ruled out by Chen, Deng, and Teng~\cite{chen2009settling}. Specifically, they showed that $1/\poly(n)$-{\em approximate} Nash equilibrium is hard, which by Spielman and Teng's proposition rules out any smoothed efficient algorithms for noise magnitude $1/\poly(n)$ (assuming $\PPAD$ is not contained in search-\RP).
Chen, Deng, and Teng nevertheless conjectured that for constant magnitude noise, two-player Nash equilibrium should have a polynomial time algorithm. 

Progress on smoothed complexity of Nash with $\varepsilon$-noise (for small constant $\varepsilon>0$) was made by~\cite{rubinstein2016settling} who proved the following hardness of approximation result:
assuming the ``Exponential Time Hypothesis for \PPAD%
\footnote{The Exponential Time Hypothesis (ETH) for \PPAD~is a strengthening of $\PPAD \nsubseteq \text{(search-)}\RP$, which postulates that End-of-Line (the canonical \PPAD-problem) requires $2^{\tilde{\Omega}(n)}$ time.}
'', 
finding an $\varepsilon$-approximate Nash equilibrium requires quasipolynomial ($\approx n^{\log(n)}$) time.
By Spielman and Teng's proposition, this hardness of approximation result also implies an analogous quasipolynomial hardness in the smoothed setting.
For hardness of approximation, the result of~\cite{rubinstein2016settling} is essentially optimal due to a matching quasipolynomial time approximation algorithm~\cite{lipton2003playing}.
This quasipolynomial time algorithm does not extend to the smoothed case, and a large gap in the complexity of constant-smoothed Nash (quasipolynomial vs exponential) remained open.

In this work, we resolve the complexity of two-player Nash equilibrium with constant-magnitude smoothing, proving that it is \PPAD-complete (under randomized reductions).
Compared to~\cite{rubinstein2016settling}, we rule out smoothed polynomial time algorithms under a much weaker assumption ($\PPAD \nsubseteq \text{(search-)}\RP$ vs ETH for \PPAD).\footnote{As discussed above, this holds in the case that $X$ is approximately polynomial-time sampleable -- otherwise we require the assumption $\PPAD \nsubseteq \text{(search-)}\RP/\poly$.} Alternatively, comparing both results under the same assumption,
ETH for \PPAD, we prove a much stronger lower bound on the running time ($2^{\poly(n)}$ vs $\approx n^{\log(n)}$)%
\footnote{In fact, under the plausible hypothesis that the true complexity of End-of-Line (the canonical \PPAD-problem) is $\approx 2^{n^\alpha}$ for some constant $0< \alpha \leq 1/2$, our result implies the qualitatively-same strong lower bound, {\em and} the result of ~\cite{rubinstein2016settling} completely breaks.}.
Finally, another advantage of our result compared to~\cite{rubinstein2016settling} is that our proof is much simpler, and in particular does not require any PCP-like machinery.

\subsection{Intuition and roadmap}

We will reduce $1/\poly(n)$-approximate Nash to $X$-SMOOTHED-NASH.
The starting point of our reduction is the following simple idea: for any mixed strategies $(x,y)$ which are spread over a large number of actions, the noise from the smoothing averages out.
In contrast, if we start with an off-the-shelf \PPAD-hard game $(P,Q)$ and amplify it by simple repetition (formally, tensor the payoff matrices $P,Q$ with the all ones matrix $J$), the signal from $P,Q$ will remain strong even with respect to well-spread strategies.
This means that given a well-spread (in a sense we make precise later) Nash equilibrium $x,y$ for a tensored, smoothed game $(P \otimes J + N_P, Q \otimes J + N_Q)$, we can recover a $1/\poly(n)$-approximate equilibrium for $(P,Q)$.

There is one major problem with the reduction suggested above: an oracle for $X$-SMOOTHED-NASH might not return a well-spread equilibrium $(x,y)$.
Our goal henceforth is to modify this construction to create a game where no Nash equilibrium has strategies concentrated on a small number of actions.
Note that pure or even small-support equilibria don't break only our proof approach: they can be found efficiently by brute-force enumeration, so such games cannot be hard.

Which games have no strategies concentrated on a small number of actions?
At one extreme, if the entries of the payoff matrices are entirely i.i.d.~(from any continuous distribution), a folklore result states that the game has a pure equilibrium with probability approaching $1-1/e$.
This creates a significant problem: we have to work with games where the entries are smoothed with independent noise -- if such games turn out also to have pure or small-support strategies, then they cannot be hard.

In contrast to i.i.d.~random games, we observe that random zero-sum games tend to have only well spread equilibria~\cite{Rob06,jonasson2004optimal}.
For example, they are exponentially unlikely to have a pure equilibrium; intuitively, if a pure strategy profile is exceptionally good for one player, it is likely exceptionally bad for the other.
In the context of our proof approach, another advantage of random zero-sum games is that with respect to well-spread mixed strategies, they will also average out.
That is, even if we add a random zero-sum game $Z$, we can still hope to recover a $1/\poly(n)$ Nash equilibrium for $(P,Q)$ from a well-spread equilibrium for $(P \otimes J + Z + N_P, Q \otimes J - Z + N_Q)$.
Our main technical task is to show that adding a random zero-sum game in this fashion produces a game with only well-spread Nash equilibria, even in the presence of the i.i.d. smoothing $N_A, N_B$.

\paragraph{Our first step is to rule out all small support equilibria.}
In Section~\ref{sec:large-support} we formalize the above intuition, showing that every equilibrium of a random zero-sum game has large supports, even when we add constant-magnitude perturbations.
For technical reasons, our proof in this section works for random zero-sum games whose entries are drawn uniformly from discrete $\{-1,1\}$.

\paragraph{Our second step is to obtain a robust version of no-small-support.}
Namely, building on the fact every equilibrium has large support, in Section~\ref{sec:small-norm} we prove that it must be well-spread (formally, the mixed strategies have small $||\cdot||_2$ norm).
For technical reasons, our proof in this section works for random zero-sum games whose entries are drawn uniformly from continuous $[-1,1]$.
Fortunately, we can make both of proofs work simultaneously by taking the sum of a $\{-1,1\}$ and a $[-1,1]$ zero-sum games.

\paragraph{Putting it all together.}
To summarize, our final construction of hard instance is given by:
\begin{gather*}
\begin{array}{r@{\ }c@{\ }c@{\ }c@{\ }c@{\ }l}
A   := & P \otimes J & + &Z_{\{-1,1\}} &+&\,\,\, Z_{[-1,1]}\\
B  := & \underbrace{Q \otimes J}_{\text{\PPAD-hard}} &-&  \underbrace{Z_{\{-1,1\}}}_{\text{large support}} &-  &\underbrace{Z_{[-1,1]}}_{\text{well-spread}} ,
\end{array}
\end{gather*}
where $Z_{\{-1,1\}},Z_{[-1,1]}$ are random matrices with i.i.d.~entries uniformly sampled from $\{-1,1\}$ and $[-1,1]$ (respectively), and $(P,Q)$ is a \PPAD-hard bimatrix game, and $J$ is an (appropriate-dimension) all-ones matrix.

In Section~\ref{sec:PPAD}, we show that when the Nash equilibrium strategies are well-spread, the random zero-sum games and random perturbations average out.
Thanks to the amplification, the signal from $(P,Q)$ remains sufficiently strong.
Thus, we can map any Nash equilibrium of $(A,B)$ to a $1/\poly(n)$-approximate Nash equilibrium of $(P,Q)$. By~\cite{chen2009settling} this suffices to establish \PPAD-hardness (under randomized reductions). 

\begin{remark*}[Inverse-polynomial signal-to-noise ratio] \hfill

Interestingly, the amplification of $(P,Q)$ by repetition is so powerful that our proof would go through even if we were to multiply $P$ and $Q$ by an inverse-polynomial small scalar.%
\footnote{We only informally state the result to prioritize simplicity, but it will be evident by the remarks in Section~\ref{subsec:remarks}.}
In this sense, we show that Nash remains intractable even subject to noise (zero-sum + i.i.d.) that is {\em polynomially larger} than the worst-case signal. 
\end{remark*}

\subsection{Additional related work}
Subsequent to the seminal works of~\cite{daskalakis2009complexity, chen2009settling} which showed that Nash equilibrium is \PPAD-complete, there has been an active line of work on algorithms with provable guarantees for exact or approximate equilibria in special cases including:  sparse games~\cite{barman2018approximating}, low-rank games~\cite{kannan2010games,adsul2011rank}, positive-semidefinite games~\cite{alon2013approximate}, anonymous games~\cite{daskalakis2015approximate,cheng2017playing}, tree {\linebreak}games~\cite{elkind2006nash,barman2015approximating,ortiz2016fptas}.
Complexity limitations for most of these special cases known as well: sparse games~\cite{chen2006sparse,liu2018approximation}, low-rank games~\cite{mehta2014constant}, anonymous games~\cite{chen2015complexity}, and tree games~\cite{deligkas2020tree}.

More relevant to the topic of smoothed analysis, it is known that when equilibria do not fluctuate when the input is perturbed, finding equilibria can be done efficiently~\cite{balcan2017nash}. 
Furthermore, a game chosen at random is likely to have easy-to-find equilibria~\cite{barany2007nash}. 

Spielman and Teng~\cite[Open Question 11]{spielman2006smoothed} ask whether there is a relation between approximation hardness and smoothed lower bounds: the former implies the latter, but little else is known regarding smoothed lower bounds. 
For the case of integer linear programs over the unit cube, Beier and V\"ocking~\cite{beier2006typical} show that a problem has polynomial smoothed complexity if and only if it admits a pseudo-polynomial algorithm. Note that a pseudo-polynomial algorithm can be used to approximate by truncating input numbers.
For other problems, we are aware of a few papers that argue smoothed complexity lower bounds via approximation hardness, e.g.~\cite{chen2009settling,huang2007approximation,kelner2007hardness}.%

\section{Preliminaries}\label{sec:prelim}
We formally define the problem here, and present some remarks.
Let $n$ be a positive integer. We let $\be_i \in \R^n$ be the $i$th indicator vector. Let $A, B \in \R^{n \times n}$ be payoff matrices (corresponding to Alice and Bob). We define a \emph{Nash equilibrium} to be vectors $\bx, \by \in \R^n_{\ge 0}$, called \emph{mixed strategies}, such that $\|\bx\|_1 = \|\by\|_1 = 1$ we have that
\begin{align*}
\bx\trans A \by &= \max_{i \in [n]} \be_i\trans A \by\\
\bx\trans B \by &= \max_{i \in [n]} \bx\trans B \be_i.
\end{align*}

We say that an equilibrium is $\eps$-approximate if
 
\begin{align*}
\bx\trans A \by + \eps &\ge \max_{i \in [n]} \be_i\trans A \by\\
\bx\trans B \by + \eps &\ge \max_{i \in [n]} \bx\trans B \be_i.
\end{align*}

For a given equilibrium $\bx, \by$ (often clear from context), we let $A^{00}$ and $B^{00}$ be the restrictions of $A$ and $B$ to $\supp(\bx) \times \supp(\by)$, respectively. We let $A^{10}$ and $B^{10}$ be the restrictions to $\overline{\supp(\bx)} \times \supp(\by)$, etc.
Computing any Nash equilibrium, even $n^{-O(1)}$-approximate, is known to be \PPAD-complete~\cite{chen2009settling}:

\begin{theorem}[\cite{chen2009settling}]\label{thm:CDT}
For all $c > 0$, computing an $n^{-c}$-approximate Nash equilibrium of an $n\times n$ bimatrix game with entries bounded in $[0, 1]$ is \PPAD-complete.
\end{theorem}

\subsection{Remarks on the reduction}\label{subsec:remarks}

The reduction, presented in Section~\ref{sec:reduction}, will ultimately take a hard instance of Theorem~\ref{thm:CDT} and transform it into a instance of $X$-SMOOTHED-NASH, for suitable distributions $X$. By the nature of the reduction, if one applies the same reduction with a wider hardness-of-approximation guarantee, one can deduce that for a suitable constant $c > 0$, it is \PPAD-hard under a randomized reduction to find a $n^{-c}$-approximate equilibrium of $X$-SMOOTHED-NASH (see, e.g., Eq.~\ref{eq3}). This has two interesting implications.

First, this means that if you truncate the output of the distribution $X$, as well as the uniform distribution sampled in the reduction, to $O(\log n)$ bits, it is still \PPAD-hard to find an (approximate) equilibrium for the resulting instance . In particular, the smoothed complexity result is robust to the underlying arithmetic representation of the payoffs.

Second, scaling down the hard instance of Theorem~\ref{thm:CDT} by a small polynomial still maintains an $n^{-O(1)}$ hardness-of-approximation guarantee. Thus, as mentioned in the introduction, the reduction implies that  Nash remains intractable even subject to noise (zero-sum + i.i.d.) that is \emph{polynomially} larger than the worst-case signal.

\subsection{Concentration for random bilinear forms}\label{sec:matrix-bilinear-concentration}
\newcommand{\N}{\mathbb{N}}
\newcommand{\iprod}[1]{\langle #1 \rangle}
\renewcommand{\E}{\mathbb{E}} %

We introduce here the following concentration bound which is useful in our result.

\begin{definition}[Subgaussian random variable]
  A $\R$-valued random variable $X$ is subgaussian with variance proxy $s^2 > 0$ if for all $t > 0$, $\E \exp(tX) \leq \exp(s^2 t^2 / 2)$.
  Note that if $X \in [-b,b]$ for some $b > 0$ with probability $1$, then $X$ is subgaussian with variance proxy $b^2/4$.
\end{definition}

\begin{lemma}\label{lem:conc-subgauss}
  Let $A$ be an $n \times n$ matrix with independent subgaussian entries with variance proxy at most $1$.
  For all $u > 0$, with probability at least $1-\exp(-u^2)$, all $\bx,\by \in \R^n$ with $\|\bx\|_2 = \|\by\|_2 = 1$ have
  \[
  \bx^\top A \by \leq O(\sqrt{\log n} + u)(\|\bx\|_1 + \|\by\|_1) \, .
  \]
  As a corollary, with the same probability, all $\bx,\by \in \R^n$ with $\|\bx\|_1,\|\by\|_1 \leq 1$ have
  \[
    \bx^\top A \by \leq O(\sqrt{\log n} + u)(\|\bx\|_2 + \|\by\|_2) \, .
  \]
\end{lemma}

The proof of this lemma is deferred to the Appendix.

\section{The Reduction, and Proof of Theorem~\ref{thm:main}}\label{sec:reduction}
\label{sec:PPAD}

First, we show in Section~\ref{subsec:sym} the reduction in the case that the noise distribution $X$ is symmetric, i.e., the probability of sampling $a$ and $-a$ is identical for all $a > 0$. We then show in Section~\ref{subsec:general} a slight modification which works for any distribution $X$.

\subsection{The symmetric case}\label{subsec:sym}

Let $\eps > 0$ be a sufficiently small constant. Let $X$ be any symmetric distribution on $[-\eps, \eps]$. Let $n, b$ be positive integers
such that $b$ divides $n$, $b = n^{0.01}$, and $n$ is sufficiently large. We divide $[n]$ into $b$ \emph{blocks} which we label $I_i := \{(i-1)\tfrac n b + 1, (i-1) \tfrac n b +2 , \hdots, i \cdot \tfrac n b\}.$ We let $\ell := n / b = n^{0.99}$ denote the \emph{block length}.

Let $P, Q \in \mathbb R^{b \times b}$ be payoff matrices. Let $J_\ell$ denote the $\ell\times \ell$ all $1$'s matrix. Let $Z_0$ be an $n\times n$ matrix whose entries are sampled i.i.d.~from the Rademacher distribution (i.e., the uniform distribution on $\{-1, 1\}$). Let $Z_1$ be an $n \times n$ matrix whose entries are sampled i.i.d.~from the uniform distribution on $[-1, 1]$.  Let $A_{\eps}, B_{\eps}$ be $n\times n$ matrices whose entries are i.i.d. sampled from $X$ (all distributions independent).\footnote{The to meet the definition of $X$-SMOOTHED-NASH, which specifies that the hard game must have entries between $[-1, 1]$, we can scale the construction (and thus $X$) by a factor of $3$.}
\begin{align*}
 A &:= P \otimes J_{\ell} + Z_0 + Z_1 + A_{\eps}\\
 B &:= Q \otimes J_{\ell} - Z_0 - Z_1 + B_{\eps},
\end{align*}
where $P \otimes J_{\ell}$ denotes the $n\times n$ matrix, where every entries in block $I_i \times I_j$ is $P_{i,j}$.

\vspace{1em}
We present here here the final result of this paper. We will refer without proof to a bound on the norm of the equilibrium strategy vectors, and we defer its proof to the rest of the paper, namely Sections~\ref{sec:large-support} and~\ref{sec:small-norm}.
This norm bound is the technical heart of this paper, and the present section illustrates its strength.

We seek to show that equilibria of the reduced game $(A,B)$ can be used to efficient produce approximate equilibria to the game $(P,Q)$, which we have assumed is hard to approximate.
Let $(\bx, \by)$ be an equilibrium of $(A, B)$. We will show in Section~\ref{sec:small-norm} that, with high probability, $\|\bx\|_2, \|\by\|_2 \le n^{-0.2}$, even when $\epsilon$ is a constant. Note that  $b=n^{0.01}$ is the dimension of the input game $(P,Q)$. Define $(\hat{\bx}, \hat{\by})$ to be distributions over $[b]$ such that for all $i \in [n]$
\begin{align*}
\hat{x}_i &= \sum_{i' \in I_i} x_{i'},&
\hat{y}_i &= \sum_{i' \in I_i} y_{i'}.
\end{align*}    
\begin{theorem}
With probability $1 - n^{-2}$, we have that $(\hat{\bx}$, $\hat{\by})$ is a $b^{-19}$-approximate equilibrium of $(P, Q)$.
\end{theorem}

\begin{proof}
We claim that $(\hat{\bx}, \hat{\by})$ is an $b^{-19} = n^{-0.19}$-approximate equilibrium of $(P, Q)$ with high probability. Assume not, without loss of generality, Alice would benefit from deviating from $\hat{\bx}$. That is, there exists $i \in [b]$ such that
\begin{align}
    \hat{\bx}\trans P \hat{\by} &\le \be_i\trans P \hat{\by} - b^{-19}.\label{eq:not-apx-equ}
\end{align}
Define $\bm u_S$ to be the uniform probability vector on support $S$, then, the above is equivalent to
\begin{align}
\bx\trans (P \otimes J_{\ell}) \by \le \bu_{I_i}\trans (P \otimes J_{\ell}) \by - b^{-19}.\label{eq:not-apx-equ-tensored}
\end{align}
By Lemma~\ref{lem:conc-subgauss}, we may assume that the concentration inequality holds for $\frac{1}{2+\eps}(Z_0 + Z_1 + A_{\eps})$, then we know that
\begin{align}
|\bx\trans (Z_0 + A_{\eps}) \by| &\le O(\sqrt{\log n}\ n^{-0.2})\label{eq:(1)}\\
|\bu_{I_i}\trans (Z_0 + A_{\eps}) \by| &\le O(\sqrt{\log n}\ n^{-0.2})\label{eq:(2)}
\end{align}

Combining Eqs.~\ref{eq:not-apx-equ-tensored},\ref{eq:(1)}, and \ref{eq:(2)} we get
\begin{align}
\bx\trans A \by \le \bu_{I_i}\trans A \by - b^{-19} + O(\sqrt{\log n}\ n^{-0.2}) &< \bu_{I_i}\trans A \by.\label{eq3}
\end{align}
since $b = n^{0.01}$. This contradicts that $(\bx, \by)$ is a Nash equilibrium of $(A, B)$.

By a similar argument, Bob does not wish to deviate with high probability. Therefore, $(\hat{\bx}, \hat{\by})$ is a $b^{-19}$-approximate Nash equilibrium of $(P, Q)$.
\end{proof}

Since finding a $b^{-19}$-approximate Nash equilibrium is PPAD-hard~\cite{chen2009settling} when $P$ and $Q$ have constant sized entries, finding the smoothed equilibrium of $(A, B)$ is PPAD-hard. 
 Since the proofs of Sections~\ref{sec:large-support} and~\ref{sec:small-norm} hold when $X$ is supported on $[-\eps,\eps]$ for $\eps>0$ constant, this is an instance of $X$-SMOOTHED NASH, and therefore concludes the proof of Theorem~\ref{thm:main} when $X$ is a symmetric distribution.

\subsection{General $X$}\label{subsec:general}

Let $X$ be any distribution supported on $[-\eps/2, \eps/2]$. Let $Y := X - X'$ be the distribution on $[-\eps, \eps]$ which takes two i.i.d.~samples from $X$ and subtracts them. Note that $Y$ is a symmetric distribution, so by the previous section we have that $Y$-SMOOTHED NASH is hard. In particular, it is hard to find an equilibrium from the distribution
\begin{align*}
  A &:= P \otimes J_{\ell} + Z_0 + Z_1 + A_{Y}\\
  B &:= Q \otimes J_{\ell} - Z_0 - Z_1 + B_{Y},
\end{align*}
where $A_Y$ and $B_Y$ are matrix whose entries are i.i.d.~samples from $Y$. We can rewrite $A_Y = A_X - A'_{X}$ and $B_Y = B_X - B'_X $, where $A_X, A'_X, B_X, B'_X$ are all i.i.d.~matrix samples from $X$. Thus, the distribution can be rewritten as
\begin{align*}
  A &:= (P \otimes J_{\ell} + Z_0 + Z_1 - A'_{X}) + A_X\\
  B &:= (Q \otimes J_{\ell} - Z_0 - Z_1 - B'_{X}) + B_X,
\end{align*}
This is an instance of $X$-SMOOTHED NASH, and we conclude Theorem~\ref{thm:main} for arbitrary $X$, losing a factor 2 on $\eps$.

\section{Equilibria Have Large Support}\label{sec:large-support}
In this section and the following, we will show the bound on $\Vert \bm x\Vert_2,\, \Vert \bm y\Vert_2$ which was required in the proof of Theorem~\ref{thm:main}.
We first show that the support of the equilibria is large with high probability. Then, in Section~\ref{sec:small-norm}, use this to argue that the weight must be sufficiently spread.
The main result of this section is the following lemma.

\begin{lemma}\label{lem:whp-large-support-uniform}
With probability $1 - n^{-3}$, for every Nash equilibrium $(\bx, \by)$ of $(A, B)$, we have that $\lvert\supp(\bx)\rvert= \lvert\supp(\by)\rvert > n^{0.96}$.
\end{lemma}

We prove this result using methods partially inspired by~\cite{jonasson2004optimal}. Observe that a Nash equilibrium of $(A, B)$ requires that
\begin{align}
    \bx\trans A\by &\ge \be_i\trans A\by &&\text{for all }i\in [n]\nonumber\\
    \bx\trans B\by &\ge \bx\trans B\be_j &&\text{for all }j\in [n]\nonumber\\
    \implies \bx\trans (A + B)\by &\ge \be_i\trans A\by + \bx\trans B \be_j &&\text{for all }\label{eq:123} i, j \in [n].
\end{align}

We seek to show that Eq.~\ref{eq:123} cannot hold when the support $\bx, \by$ is sufficiently small.\footnote{In the case of~\cite{jonasson2004optimal}, which considers zero-sum games, the LHS of (\ref{eq:123}) is equal to $0$, so it suffices to bound the probability that the RHS is positive for \emph{some} $i$ and $j$.} To do that, we propose a ``benchmark'' to which both the LHS and the maximum value of the RHS of Eq.~\ref{eq:123} are comparable to. To define this benchmark, we begin by introducing a notion of {\em robust partition} of the strategy vectors.
Consider $\bx \in \mathbb R^n$ such that $\|\bx\|_1 = 1$. 
Let $L = \lceil \log_2 n\rceil / 2^{100}$. Let $D = 2^{2^{500}}$. Let $E_1, \hdots, E_{L}$ be intervals such that $E_i = (D^{-i}, D^{-(i-1)}]$ for all $i < L$ and $E_L = [0, D^{-(L-1)}]$. Let $\bx = \bx^{(1)} + \cdots + \bx^{(L)}$ such that
  \[
    \bx^{(i)}_j = \begin{cases}
      \bx_j & \bx_j \in E_i\\
      0 & \text{otherwise}
      \end{cases}
  \]

  We say that $\bx^{(i)}$ is \emph{sparse} if it has at most $L$ nonzero coordinates; otherwise we say $\bx^{(i)}$ is \emph{dense}. Let $\bx_{\spa}$ be the sum of the sparse $\bx^{(i)}$'s and $\bx_{\den}$ be the sum of the dense ones. Note that $\bx = \bx_{\spa} + \bx_{\den}$. Now define the following quantity
  \[
    \beta(\bx) = \sqrt{\log n}\|\bx_{\den}\|_2 + \|\bx_{\spa}\|_1.
  \]

  We call $\beta(\bx)$ the \emph{benchmark} for $\bx$. This quantity will appear in a number of concentration/anti-concentration inequalities. First, we show a key anticoncentration inequality concerning this robust partition.

\begin{lemma}\label{lem:rad-anti}
  Assume that $X$ is the uniform distribution on $\{-1, 1\}$ (i.e., the Rademacher distribution).  There exists a universal constant $c > 0$ with the following property: For all $\bx \in \mathbb R^n$ such that $\|\bx\|_1 = 1$, with probability at least $n^{-0.001}$ over $\bv \sim X^n$ 
  \[
    \langle \bv, \bx\rangle \ge c \beta(\bx).
  \]
\end{lemma}

The proof of the above lemma is deferred to the Appendix.
The following concentration bound will also be of use.
For any distribution $X$, we let $X^{n\times n}$ denote the distribution of $n\times n$ matrices with entries i.i.d.~samples from $X$.

\begin{claim}\label{claim:uniform-conc}
  Let $X$ be any distribution on $[-1, 1]$. There exists a universal constant $C > 0$ such that for all $n \ge 0$, with probability $1- 1/n^4$ over $M \sim X^{n \times n}$, for all $\bx, \by \in \mathbb R^n$ such that $\|\bx\|_1 = \|\by\|_1 = 1$, we have that
  \[
    |\bx\trans M \by| \le C\cdot (\beta(\bx) + \beta(\by)).
  \]
\end{claim}

\begin{proof}
  Apply Lemma~\ref{lem:conc-subgauss} to $M$ with $u = \sqrt{3\log n}$. Then, there is a universal constant $C'$ such that with probability $1 - 1/n^3$, for all $\bx, \by$ with $\ell_1$ norm $1$, 
\[
  |\bx_{\den}\trans M \by_{\den}| \le C'\sqrt{\log n} (\|\bx_{\den}\|_2 + \|\by_{\den}\|_2).
\]
Thus, since the entries of $M$ have absolute value at most $1$,
\begin{align*}
  |\bx \trans M \by| &\le |\bx \trans M \by_{\spa}| + |\bx_{\spa} \trans M \by_{\den}| + |\bx_{\den}\trans M \by_{\den}|\\
                     &\le \|\by_{\spa}\|_1 + \|\bx_{\spa}\|_1 + C'\sqrt{\log n} (\|\bx_{\den}\|_2 + \|\by_{\den}\|_2)\\
                     &\le \max(C', 1) (\beta(\bx) + \beta(\by)).
\end{align*}

Thus, we can set $C = \max(C', 1)$.
\end{proof}

These lemmas will allow us to prove Lemma~\ref{lem:whp-large-support-uniform}. We present first the following facts about equilibria in random games.

\begin{proposition}\label{prop:same-support}
With probability $1$, for nonempty $S, T \subset [n]$ there is at most one Nash equilibrium $(\bx, \by)$ of $(A, B)$ with $S = \supp(x)$ and $T = \supp(y)$. Further, with probability $1$ all such equilibria have $|S| = |T|$.
\end{proposition}

\begin{proof}
Fix nonempty $S, T \subset [n]$. Fix $i_0 \in S$. Assume without loss of generality that $|S| \ge |T|$.
Denote $A^{00}$ as the sub-matrix of $A$ restricted to rows indexed by $S$ and columns indexed by $T$. For any equilibrium $(\bx, \by)$ with supports $S$ and $T$, we have that $\bx\trans A \by = \be_i\trans A^{00} \by$ for all $i \in S$, when treating $\bx$ and $\by$ as $|S|$- and $|T|$-dimensional vectors, respectively. Therefore, 
\begin{align}(\be_i - \be_{i_0})\trans A^{00} \by = 0 \text{ for all $i \in S \setminus \{i_0\}$}.\label{eq:space}\end{align} Since all the entries of $A^{00}$ are drawn independently from a continuous distribution, the null space of the linear system (\ref{eq:space}) has dimension $\max(|T|-|S|+1, 0) \le 1$ with probability $1$. Since $\by \neq 0$ the null space must have dimension exactly $1$. Thus, $|T| - |S| + 1 \ge 1$, which implies that $|S| = |T|$ and the solution $\by$ is unique, as there can be at most one vector in a $1$-dimensional subspace with coordinates summing to one. By a similar argument $\bx$ is also unique.  

Since there are only finitely many choices of $S$ and $T$, with probability $1$ the proposition holds for all Nash equilibria simultaneously.
\end{proof}

With probability $1$, all equilibria of $A$ and $B$ will have the same support size, and further, for every pair of possible supports $S \subset [n]$ and $T \subset [n]$ there is at most one equilibrium. We let $\bx, \by \in \mathbb R^n$ denote the probability distributions of strategies in this equilibrium.

We can now prove Lemma~\ref{lem:whp-large-support-uniform}

\begin{proof}[Proof of Lemma~\ref{lem:whp-large-support-uniform}.]
Assume (which happens with probability $1 - n^{-4}$) that the event described in Claim~\ref{claim:uniform-conc} occurs for $M = \frac{1}{2\eps} (A_{\eps} + B_{\eps})$. Fix $S, T \subset [n]$ with $|S|, |T| < \ell/10$. We seek to show that with probability at most $2^{-\ell}$, $S$ and $T$ can be the support of a Nash equilibrium. By Proposition~\ref{prop:same-support}, we can assume that $|S| = |T|$.

Also by Proposition~\ref{prop:same-support}, with probability $1$, there is at most one equilibrium $(\bx, \by)$ on the game $(A^{00}, B^{00})$ with full support. Note that $\bx$ and $\by$, if they exist, are independent of the entries of $A$ and $B$ outside of $S \times T$. As mentioned earlier in the section, in order for the equilibrium to extend, the Ineq.~\ref{eq:123} must hold:

\begin{align*}
    \bx\trans (A + B)\by \ge \be_i\trans A\by + \bx\trans B \be_j&&\text{for all } i, j \in [n].
\end{align*}

Say that $i \in [n] \setminus S$ is \emph{$S$-good} if $\be_i\trans (Z_0 + Z_1 + A_{\eps}) \by > c \beta(\by)$. By Lemma~\ref{lem:rad-anti}, we know that $\be_i\trans Z_1 \by > c\beta(\by)$ with probability at least $n^{-0.001}$. Independently, we have that $\be_i\trans (Z_0 + A_{\eps})\by \ge 0$ with probability at least $1/2$ (since $Z_0 +A_\eps$ is a mean-zero matrix distribution). Therefore, both this event happens with probability at least $n^{-0.001} / 2 \ge n^{-0.01}$.

Likewise, say that $j \in [n] \setminus T$ is \emph{$T$-good} if $\bx\trans (-Z_0 -Z_1 - B_{\eps}) \be_j > c \beta(\bx)$. By the same argument, this also happens with probability at least $n^{-0.01}$ . Furthermore, the $S$-good events and $T$-good events are independent of each other because each event is based on a disjoint subset of entriesZ from $Z_0$ and $Z_1$.

Since $\bx$ and $\by$ are probability distributions, there exists $i_0 \in S$ and $j_0 \in T$ such that $\be_{i_0}\trans (P \otimes J_{\ell})\by \ge \bx\trans (P\otimes J_{\ell})\by$ and $\bx\trans (Q\otimes J_{\ell})\be_{j_0} \ge \bx\trans (Q \otimes J_{\ell})\by$. Let $i', j' \in [b]$ be the indices of the blocks such that $i_0 \in I_{i'}$ and $j_0 \in I_{j'}$. Since we assume that $|S|, |T| \le \ell/10$, we have that $I_{i'} \setminus S$ and $I_{j'} \setminus T$ both have size at least $9\ell/10$.

Now, for any good $i \in I_{i'} \setminus S$ and good $j \in I_{j'} \setminus T$, we have
\begin{align*}
    \bx\trans (A + B)\by &= \bx\trans (P \otimes J_{\ell}) \by + \bx\trans (Q \otimes J_{\ell})\by + \bx\trans (A_{\eps} + B_{\eps})\by\\
    &\le \be_{i_0}\trans (P \otimes J_{\ell})\by + \bx\trans (Q \otimes J_{\ell})\be_{j_0} + 2C \eps(\beta(\bx) + \beta(\by))\\
    &= \be_{i}\trans (P \otimes J_{\ell})\by + \bx\trans (Q \otimes J_{\ell})\be_{j} + 2C \eps (\beta(\bx) + \beta(\by))\\
    &< \be_{i}\trans (P \otimes J_{\ell})\by + \bx\trans (Q \otimes J_{\ell})\be_{j} + c(\beta(\bx) + \beta(\by))\qquad\qquad\qquad\qquad\text{        ($\eps < 2c / C$)}\\
    &< \be_{i}\trans (P \otimes J_{\ell} + Z_0 + Z_1 + A_{\eps})\by + \bx\trans (Q \otimes J_{\ell} - Z_0 - Z_1 + B_{\eps})\be_{j}\\
    &= \be_{i}\trans A\by + \bx\trans B\by,
\end{align*}
which contradicts Ineq.~\ref{eq:123}. Thus, there must either be no good $i \in I_{i'} \setminus S$ or there is no good $j \in I_{j'} \setminus T$. This happens with probability at most
\[
2\left(1-n^{-0.01}\right)^{9\ell/10} \le 2e^{-(0.9)\ell/n^{0.01}} \le e^{-n^{0.97}},
\]
where we use in the last inequality that $n$ is sufficiently large. The number of pairs $S, T$ with support at most $n^{0.96}$ is at most
\[
\binom{n}{\le n^{0.96}}^2 \le n^{2n^{0.96}}. 
\]

Note that for $n$ sufficiently large, $n^{2n^{0.96}} e^{-n^{0.97}} \ll n^{-4}$. Thus, all equilibria have support size greater than $n^{0.96}$ with probability at least $1 - 2n^{-4}\ge 1-n^{-3}$.
\end{proof}

\section{Equilibria Have Small $\ell_2$ norm}\label{sec:small-norm}
Towards showing the missing bound in the proof of Theorem~\ref{thm:main}, the previous section showed that with high probability, any equilibrium must have polynomially large support. 
We complete here the proof of the norm bound, which in turn completes the proof of Theorem~\ref{thm:main}.

\begin{lemma}%
\label{lem:whp-small-2-norm-unif}
With probability $1 - 20n^{-3}$, for every Nash equilibrium $(\bx, \by)$ of $(A, B)$, we have that $\|\bx\|_2, \|\by\|_2 \le n^{-0.2}$.
\end{lemma}

We must, however, begin this section with a few technical results.
We will need the following theorem, which is derived from the fact that the VC-dimension of the set of halfspaces in $\R^d$ has VC-dimension at most $d+1$ -- that is, the VC-dimension of $\{ \bx \mapsto 1[\ip{\bx,\bv} + t \geq 0] \, : \, \bv \in \R^d, t \in R \}$ is at most $d+1$. (See e.g. \cite{wainwright2019high}, Example 4.21.)

\begin{theorem}[Multivariate Glivenko-Cantelli]
  \label{thm:glivenko-cantelli}
  Let $X$ be a random vector in $\R^d$ and let $X_1,\ldots,X_n$ be independent copies of $X$.
  For all $\delta \in [0,1]$, with probability $1-\delta$,
  \[
  \sup_{\bv \in \R^d, t \in \R} \left | \frac 1 n \sum_{i =1}^n 1[\ip{X_i,\bv} \geq t] - \Pr_X (\ip{X,\bv} \geq t) \right | \leq O\left ( \sqrt{\frac d n} + \sqrt{\frac{\log(1/\delta)} n} \right ) \, .
  \]
\end{theorem}

We also need the following Littlewood-Offord-type theorem.

\begin{theorem}[\cite{rudelson2015small}, Theorem 1.2]\label{thm:small-ball-sums}
  Let $X_1,\ldots,X_n$ be real-valued independent random variables with densities almost everywhere bounded by $K$.
  Let $a_1,\ldots,a_n \in \R$ with $\sum_{i \leq n} a_i^2 = 1$.
  Then the density of $\sum_{i \leq n} a_i X_i$ is bounded by $\sqrt{2} K$ almost everywhere.
\end{theorem}

The following lemma, which we obtain as a corollary of these two theorems, allows us to argue that the entries of a product of a random matrix with a fixed vector are relatively spread out. 

\begin{lemma}\label{lem:half-space-game}
Let $n, d$ be positive integers. %
Let $X$ be an $\R$-valued random variable with density bounded by $K$.
Let $\bg_1\hdots, \bg_n$ be independent random vectors in $\R^d$ whose coordinates are independent copies of $X$.
With probability $1-\delta$, for all unit vectors $\bv \in \mathbb R^d$ and all intervals $[a, b] \subset \mathbb R$,
\[
\frac 1 n \sum_{i=1}^n 1[\ip{\bg_i,\bv} \in [a,b]] \leq \sqrt{2} K|a-b| + O\left ( \sqrt{\frac d n} + \sqrt{\frac{\log(1/\delta)} n} \right ) \, .
\]
\end{lemma}
\begin{proof}
  By Theorem~\ref{thm:glivenko-cantelli}, with probability at least $1-\delta$, the CDFs of $\ip{\bg, \bv}$ and the empirical distribution of $\ip{\bg_i,\bv}$ have distance at most $O\left ( \sqrt{\frac d n} + \sqrt{\frac{\log(1/\delta)} n} \right )$, for all $\bv \in \R^d$.
  So it suffices to show that for every unit $\bv \in \R^d$, $\Pr_{\bg}(\ip{\bg,\bv} \in [a,b]) \leq \sqrt{2} K |a-b|$.
  This follows immediately from Theorem~\ref{thm:small-ball-sums}.
\end{proof}

Finally, this lemma allows us to prove the following claim.

\begin{claim}\label{claim:spread-unif}
Let $X$ be a distribution on $[-1, 1]$ whose probability density is at most $100$ everywhere. Let $M \sim X^{n\times n}$. With probability $1 - n^{-4}$, for every $S, T \subset[n]$ with $|S| \ge n^{0.95}$ and $|T| \le n^{0.85}$, there exists disjoint $S_1, S_2 \subset S$ of size at least $n^{0.94}$ each such that for all unit vectors $\by \in \mathbb R^n$ with support in $T$ there exists $r \in \mathbb R$ such that
\begin{align*}
\be_{i_1}\trans M\by &\ge r + n^{-0.07}&\text{ for all $i_1 \in S_1$}\\
\be_{i_2}\trans M\by &\le r&\text{ for all $i_2 \in S_2$}.
\end{align*}
\end{claim}

\begin{proof}
  For every $T \subset [n]$ of size at most $n^{0.85}$, apply Lemma~\ref{lem:half-space-game} to the rows of $M$ restricted to the columns of $T$ (so $d = |T| \le n^{0.85}$) with $\delta = e^{-n^{0.86}}$. Thus, with probability $1 - e^{-n^{0.86}}$,  for every unit vector $y \in \mathbb R^d$ supported on $T$ and every interval $[a, b]$ of length $n^{-0.06} / 10$, the number of $i \in [n]$ such that $\be_i \trans M\by \in [a,b]$ is at most
  \[
    n\left[100\sqrt{2}|a - b| + O\left(\sqrt{\frac{d}{n}} + \sqrt{\frac{\log(1/\delta)}{n}}\right)\right] = O(n^{0.94})
  \]
  choices of $i \in [n]$ for which $\be_i\trans A \by$ falls in that interval. Since $|S| \ge n^{0.95}$, this implies there exist $r \in \mathbb R$, and disjoint $S_1, S_2 \subset S$ of size at least $n^{0.94}$ such that
\begin{align*}
\be_{i_1}\trans M\by &\ge r + \frac{n^{-0.06}}{10} \ge r + n^{-0.07}&\text{ for all $i_1 \in S_1$}\\
\be_{i_2}\trans M\by &\le r&\text{ for all $i_2 \in S_2$}.
\end{align*}
Taking the union bound over all choices of $T$ we get this all happens with probability at most
\[
1 - \binom{n}{\le n^{0.85}}e^{-n^{0.86}} \ge 1 - e^{-n^{0.85}} \ge  1 - n^{-4}.\qedhere
\]
\end{proof}

We can now prove Lemma~\ref{lem:whp-small-2-norm-unif}.

\begin{proof}[Proof of Lemma~\ref{lem:whp-small-2-norm-unif}]
With probability $1 - n^{-3}$, by Lemma~\ref{lem:whp-large-support-uniform}, for every equilibrium $(\bx, \by)$ of $(A, B)$ with support $S$ and $T$, respectively, we have that $|S|=|T| \ge n^{0.96}$. Since there are $n^{0.01}$ blocks. By the pigeonhole principle there exists $i_0, j_0 \in [b]$ such that $|S \cap I_{i_0}|, |T \cap I_{j_0}| \ge n^{0.95}$. 

With probability $1 - 2n^{-4}$, Claim~\ref{claim:spread-unif} holds with for both $M = \frac{1}{2+\eps}(Z_0 + Z_1 + A_{\eps})$ and $M = \frac{1}{2+\eps}(-Z_0 - Z_1 + B_{\eps})$. Further, with probability at least $1-2n^{-3}$, Lemma~\ref{lem:conc-subgauss} holds for $M = \frac{1}{2+\eps}(Z_0 + Z_1 +  A_{\eps})$ and $M = \frac{1}{2+\eps}(-Z_0-Z_1+B_{\eps})$ with $u = \sqrt{3}\log n$.

We seek to show that any large-support equilibrium also has small $\ell_2$ norm. %
Assume for sake of contradiction (and without loss of generality) that $\|\by\|_2 \ge n^{-0.2}$. Let $S' = S \cap I_{i_0}$ and $T'$ be the set of coordinates of $\by$ which are greater than $n^{-0.85}$. Clearly $|T'| \le n^{0.85}$. Let $\by_{T'}$ be the coordinates of $\by$ supported on $T'$ and $\bar{\by}_{T'}$ be the remaining coordinates.  Observe that
\begin{align}
\|\bar{\by}_{T'}\|_2^2 &\le n \cdot (n^{-0.85})^2 = n^{-0.7} \le \frac{\|\by\|_2^2}{2}\label{eq:baryu}\\
\|\by_{T'}\|_2^2 &= \|\by\|_2^2 - \|\bar{\by}_{T'}\|_2^2 \ge \frac{\|\by\|_2^2}{2}.\label{eq:yu}
\end{align}
Applying Claim~\ref{claim:spread-unif} for $M = \frac{1}{2+\eps}(Z_0 + Z_1 + A_{\eps})$ and the sets $S', T'$ and the vector $\by' := \frac{\by_{T'}}{\|\by_{T'}\|_2}$, there exists $S'_1, S'_2 \in S'$ and $r \in \mathbb R$ such that (scaling by $2+\eps \ge 1$)
\begin{align*}
\be_{i_1}\trans (Z_0+Z_1+A_{\eps})\by'&\ge r + n^{-0.07}&\text{ for all $i_1 \in S'_1$}\\
\be_{i_2}\trans (Z_0+Z_1+A_{\eps})\by' &\le r&\text{ for all $i_2 \in S'_2$}.
\end{align*}
Thus,
\begin{align*}
\bu_{S'_1}\trans (Z_0+Z_1+A_{\eps})\by'&\ge r + n^{-0.07} \\
\bu_{S'_2}\trans (Z_0+Z_1+A_{\eps})\by' &\le r\\
\implies (\bu_{S'_1} - \bu_{S'_2})\trans (Z_0 + Z_1 + A_{\eps}) \by' &\ge n^{-0.07}.
\end{align*}
Applying (\ref{eq:yu}),
\[
(\bu_{S'_1} - \bu_{S'_2})\trans (Z_0 + Z_1+ A_{\eps}) \by_{T'} \ge n^{-0.07}\|\by\|_2/2 \ge n^{-0.28}.
\]
Since Lemma~\ref{lem:conc-subgauss} holds for $M = \frac{1}{2+\eps}(Z_0 + Z_1 +  A_{\eps})$, we have that
\begin{align*}
  (\bu_{S'_1} - \bu_{S'_2})\trans (Z_0 + Z_1 + A_{\eps}) \bar{\by}_{T'} &\ge -(2+\eps)C'\sqrt{\log n}(\|\bu_{S'_1} - \bu_{S'_2}\|_2 + \|\bar{\by}_{T'}\|_2)\\
&\ge -n^{0.01}\max(\sqrt{2}n^{-0.94/2}, n^{-0.7/2})\\
&\ge -n^{-0.34}.
\end{align*}
Therefore, since $\by = \by_{T'} + \bar{\by}_{T'}$
\[
(\bu_{S'_1} - \bu_{S'_2})\trans (Z_0 + A_{\eps}) \by \ge n^{-0.28} - n^{-0.34} \ge n^{-0.29}.
\]
Since $S'_1$ and $S'_2$ are subsets of the same block, we have that $\bu_{S'_1}(P \otimes J_{\ell}) = \bu_{S'_2}(P \otimes J_{\ell})$. Therefore,  
\[
(\bu_{S'_1} - \bu_{S'_2})\trans A \by \ge n^{-0.29}.
\]
But, since $S'_1$ and $S'_2$ are subsets of the support of $\bx$, we know that
\[
(\bu_{S'_1} - \bu_{S'_2})\trans A \by = 0,
\]
thus we have a contradiction. Therefore, $\|\by\|_2 \le n^{-0.2}$. By a similar argument (also with probability $1 - 5n^{-3}$, $\|\bx\|_2 \le n^{-0.2}$, as desired. By the union bound, the total probability of success is at least $1-20n^{-3} \ge 1-n^{-2}$.
\end{proof}

\section*{Acknowledgements}
We thank anonymous reviewers for helpful suggestions which improved this manuscript.

\bibliography{main}
\bibliographystyle{alpha}

\appendix

\section{Proof of Lemma~\ref{lem:conc-subgauss}}

To prove this lemma, we rely on the following powerful comparison inequality of Talagrand.

\begin{theorem}[Talagrand's comparison inequality, high-probability version. \cite{vershynin2018high}, Exercise 8.6.5]
\label{thm:talagrand-comparison}
  Suppose that $\{X_{\bs}\}_{\bs \in S}$ is a collection of $\R$-valued random variables, indexed by some $S \subseteq \R^n$, $0 \notin S$.
  Suppose that for all $\bs,\bt \in S$, $X_{\bs} - X_{\bt}$ is subgaussian with variance proxy at most $\|\bs - \bt\|_2$.
  There is a universal constant $C > 0$ such that for all $u > 0$, with probability at least $1 - \exp(-u^2)$,
  \[
  \sup_{\bs \in S} X_{\bs} \leq C \left ( \underset{\bg \sim \cN(0,I)}{\E} \sup_{\bs \in S} \iprod{\bg,\bs} + u \cdot \sup_{\bs \in S} \|\bs\|_2 \right ) \, .
  \]
\end{theorem}

Now we can prove Lemma~\ref{lem:conc-subgauss}.\\

\noindent\textbf{Lemma~\ref{lem:conc-subgauss}.} {\itshape{
  Let $A$ be an $n \times n$ matrix with independent subgaussian entries with variance proxy at most $1$.
  For all $u > 0$, with probability at least $1-\exp(-u^2)$, all $\bx,\by \in \R^n$ with $\|\bx\|_2 = \|\by\|_2 = 1$ have
  \[
  \bx^\top A \by \leq O(\sqrt{\log n} + u)(\|\bx\|_1 + \|\by\|_1) \, .
  \]
  As a corollary, with the same probability, all $\bx,\by \in \R^n$ with $\|\bx\|_1,\|\by\|_1 \leq 1$ have
  \[
    \bx^\top A \by \leq O(\sqrt{\log n} + u)(\|\bx\|_2 + \|\by\|_2) \, .
  \]
}}

\begin{proof}
  Consider for each $\bx,\by \in \R^n$ the random variable $\bx^\top A \by / (\|\bx\|_1 + \|\by\|_1)$.
  Since the entries of $A$ are subgaussian with variance proxy $1$, there is a universal $C > 0$ such that $\langle U,A \rangle$ is subgaussian with variance proxy $C \|U\|_F^2$, where $\|\cdot\|_F$ is the Frobenius norm, for any $n \times n$ matrix $U$.
  Hence, for $\bx,\by,\bx',\by' \in \R^n$,
  \[
  \frac{\bx^\top A \by}{\|\bx\|_1 + \|\by\|_1} - \frac{(\bx')^\top A \by'}{\|\bx'\|_1 + \|\by'\|_1}
  \]
  is subgaussian with variance proxy $C \|\bx\by^\top / (\|\bx\|_1 + \|\by\|_1) - (\bx')(\by')^\top / (\|\bx'\|_1 + \|\by'\|_1) \|_F^2$.
  We claim that
  \[
  \left \| \frac{\bx\by^\top}{\|\bx\|_1 + \|\by\|_1} - \frac{(\bx')(\by')^\top} {\|\bx'\|_1 + \|\by'\|_1} \right \|_F^2 \leq \left \|\frac{(\bx,\by)} {\|\bx\|_1 + \|\by\|_1} - \frac{(\bx',\by')}{\|\bx'\|_1 +\|\by'\|_1} \right \|_2^2 \, ,
  \]
  where $(\bx,\by)$ denotes the concatenation of $\bx$ and $\by$ to a $2n$-length vector.
  To see this, recalling that $\|\bx\|_2 = \|\by\|_2 = \|\bx'\|_2 = \|\by'\|_2 = 1$, let $m = \|\bx\|_1 + \|\by\|_1$ and $m' = \|\bx'\|_1 + \|\by'\|_1$ and expand both sides, it is equivalent to prove
  \[
     \frac{m^2 + (m')^2 - 2m(m')\langle \bx, \bx'\rangle \langle \by, \by'\rangle}{m^2(m')^2} \le \frac{2m^2 + 2(m')^2 - 2m(m')\langle \bx, \bx'\rangle - 2m(m')\langle \by, \by'\rangle}{m^2(m')^2}.
  \]
  This is equivalent to
  \[
    \frac 1 {m^2} + \frac 1 {(m')^2} - 2 \frac{\langle \bx,\bx' \rangle + \langle \by,\by' \rangle - \langle \bx,\bx' \rangle \langle \by',\by \rangle}{mm'} \geq 0 \, .
  \]
  Dividing by $2/mm'$ and using $1/m^2 + 1/(m')^2 \geq 2/mm'$, it is enough to show
  \[
  1 - \langle \bx,\bx' \rangle - \langle \by,\by' \rangle - \langle \bx,\bx' \rangle \langle \by',\by \rangle \geq 0 \, .
  \]
  This factors as $(1- \langle \bx,\bx' \rangle)(1 - \langle \by,\by' \rangle) \geq 0$ since we assumed $\bx,\bx',\by,\by'$ were unit vectors.

  Now we can apply Theorem~\ref{thm:talagrand-comparison} to see that with probability at least $1-\exp(-u^2)$, 
  \[
  \sup_{\substack{\bx,\by\\\|\bx\|_2=\|\by\|_2=1}} \frac{\bx^\top A \by}{\|\bx\|_1 + \|\by\|_1} \leq C \left ( \underset{\bg \sim \cN(0,I)}{\E} \sup_{\substack{\bx,\by\\\|\bx\|_2=\|\by\|_2=1}} \frac{\langle (\bx,\by),\bg \rangle}{\|\bx\|_1 +\|\by\|_1} + u \right )
  \]
  where $C$ is a universal constant, $\bg$ is a length $2n$ Gaussian vector with independent coordinates, and we have used that $\|(\bx,\by)\|_2 \leq \|(\bx,\by)\|_1 = \|\bx\|_1 + \|\by\|_1$.
  To finish the argument, observe that
  \[
  \underset{g \sim \cN(0,I)}{\E} \sup_{\substack{\bx,\by\\\|\bx\|_2=\|\by\|_2=1}} \frac{\langle(\bx,\by),\bg \rangle}{\|\bx\|_1 + \|\by\|_1} = \underset{\bg \sim \cN(0,I)}{\E} \|\bg\|_\infty \leq O(\sqrt{\log n}) \, .
  \]

  Finally, to prove the corollary, note that we just showed that with probability at least $1-\exp(-u^2)$, all $\bx,\by \in \R^n$ with $\|\bx\|_1 = \|\by\|_1 = 1$ have $\bx^\top A \by / \|\bx\|_2 \|\by\|_2 \leq O(\sqrt{\log n} + u) \cdot (1/\|\bx\|_2 + 1/\|\by\|_2)$.
  Multiplying by $\|\bx\|_2 \|\by\|_2$ implies the corollary.
\end{proof}

\section{Proof of Lemma~\ref{lem:rad-anti}}

\subsection{Facts about the binomial distribution}

In our result, we need the following bound of Erd{\H o}s~\cite{erdos1945lemma}.

\begin{theorem}[\cite{erdos1945lemma}, variant of~\cite{dzindzalieta2014tight}]\label{thm:erdos-lo}
  Let $a_1, \hdots, a_n \ge 1$ be real numbers and $\eps_1, \hdots, \eps_n$ be Rademacher random variables (uniform distribution on $\{-1, 1\}$) then for all integers $k \ge 1$,
  \[
    \Pr[a_1\eps_1 + \cdots a_n \eps_n \ge k-1] \ge \Pr[\eps_1 + \cdots + \eps_n \ge k].
  \]
\end{theorem}

Furthermore, the following binomial inequality will be useful:

\begin{lemma}[\cite{Ash1990}]\label{lem:entropy}
  For all $k$ and $n$, \[
    \binom{n}{k} \ge \frac{2^{nH(k/n)}}{\sqrt{8n}},
  \]
  where $H(\cdot)$ is the binary entropy function. 
\end{lemma}

Note that when $k = \frac{n}{2}(1 + \delta)$, then
\begin{align*}
  H(k/n) &:= -\frac{1+\delta}{2}\log_2(\tfrac12(1+\delta)) - \frac{1-\delta}{2}\log_2(\tfrac12(1-\delta))\\
             &\ge 1 -\frac1{\ln 2}%
             \left(\frac{1+\delta}{2}\cdot \delta + \frac{1-\delta}{2}\cdot (-\delta)\right)\\
             &= 1-(\log_2 e)\delta^2.
\end{align*}

Combining with the above inequality gives
\[
  \frac{1}{2^n}\binom{n}{k} \ge \frac{1}{\sqrt{8n}}e^{-n\delta^2}.
\]

This allows us to show the following:
\begin{claim}\label{claim:binom-sum}
  For all integers $n \ge k \ge 0$ with $n$ sufficiently large
  \begin{align}
    \frac 1{2^n}\sum_{i= \frac{n+k}{2}}^n \binom{n}{i} \ge \frac{1}{10000} \exp\left(-\frac{10k^2}{n}\right).\label{eq:yeah}
  \end{align}
\end{claim}

\begin{proof}
Note that here, $\delta = \tfrac kn$.
If $k\geq n-2\sqrt n$, then \[-10k^2/n\leq -10n+40\sqrt n - 400\leq -9n\] for $n$ sufficiently large. Note that the LHS of~\ref{eq:yeah} is at least $2^{-n} > e^{-9n}$, and thus is at least the RHS.

 On the other hand, if $k \le n -2 \sqrt{n}$, then by Lemma~\ref{lem:entropy}, the sum of the first $\sqrt n$ terms is at least
  \begin{align*}
    \sqrt{n} \binom{n}{\frac{n+k}{2} + \sqrt{n}} &\ge \sqrt{n} \frac{1}{\sqrt{8n}} \exp\left(-n \cdot \left(\frac{k + 2\sqrt{n}}{n}\right)^2\right)\\
    &=\tfrac 1{\sqrt 8} \exp\left(-\frac{k^2 + 4k\sqrt{n} + 4n}{n}\right)\\
    &\ge\tfrac 1{\sqrt 8} \exp\left(-\frac{5k^2 + 5n}{n}\right)\\
    &= \frac{1}{e^5\sqrt{8}}\exp\left(-\frac{5k^2}{n}\right),
  \end{align*}
  which implies the claim.
\end{proof}

We recall here the statement of Lemma~\ref{lem:rad-anti}, and give a proof with the above results:\\[6pt]

\noindent\textbf{Lemma~\ref{lem:rad-anti}.} {\itshape{Assume that $X$ is the uniform distribution on $\{-1, 1\}$ (i.e., the Rademacher distribution).  There exists a universal constant $c > 0$ with the following property: For all $\bx \in \mathbb R^n$ such that $\|\bx\|_1 = 1$, with probability at least $n^{-0.001}$ over $\bv \sim X^n$ 
  \[
    \langle \bv, \bx\rangle \ge c \beta(\bx).
  \]
 }}

\begin{proof}
Recall, we have defined the following:
let $L = \lceil \log_2 n / 2^{100}\rceil$. Let $D = 2^{2^{500}}$ and let $E_1, \hdots, E_{L}$ be intervals such that $E_i = (D^{-i}, D^{-(i-1)}]$ for all $i < L$ and $E_L = [0, D^{-(L-1)}]$. Let $\bx = \bx^{(1)} + \cdots + \bx^{L}$ such that
  $\bx_j^{(i)} = \bx_j\cdot{\mathbf 1}[\bx_j\in E_i]$.
We say $\bx^{(i)}$ is \emph{sparse} if it has at most $L$ nonzero coordinates; otherwise it is \emph{dense}. Let $F \subset \{1, 2, \hdots, L\}$ be the set of dense indices. Let $\bx_{\spa}$ be the sum of the sparse $\bx^{(i)}$'s and $\bx_{\den}$ be the sum of the dense ones, and define 
  \[
    \beta(\bx) = \sqrt{\log n}\|\bx_{\den}\|_2 + \|\bx_{\spa}\|_1.
  \]

  Note that if we drop $\bx^{(L)}$, $\beta$ changes by at most $\sqrt{\log n}\Vert \bx^{(L)}\Vert_1 \leq n\sqrt{\log n} \cdot n^{-2^{400}+1}$, a negligeably small term.
   Thus, we can without loss of generality assume that $\bx^{(L)} = 0$.

   Since $\beta(\bx) = \|\bx_{\spa}\|_1 + \sqrt{\log n}\|\bx_{\den}\|_2$, we have for any $\bx$, at least one of $\|\bx_{\spa}\|_1$ or $\sqrt{\log n}\|\bx_{\den}\|_2$ is at least $\frac{1}{2}\beta(\bx)$.
   Assume we know that with probability at least $2n^{-0.001}$, $\langle \bv, \bx_{\spa}\rangle = \Omega(\|\bx_{\spa}\|_1)$; and with probability at least $2n^{-0.001}$, $\langle \bv, \bx_{\den}\rangle = \Omega(\sqrt{\log n}\|\bx_{\den}\|_2)$. Then, we know with probability at least $n^{-0.001}$, one of $\langle \bv, \bx_{\spa}\rangle$ and $\langle \bv, \bx_{\den}\rangle$ is at least $\Omega(\beta(\bx))$ and the other is at least $0$ and thus their sum is at least $\beta(\bx)$. We split the remainder of the proof into two parts.

   \subsubsection*{Part 1, $\langle \bv, \bx_{\spa}\rangle = \Omega(\|\bx_{\spa}\|_1)$}

   Let $\bx'$ be the $2L$ largest coordinates of $\bx_{\spa}$. Note that $\|\bx'\|_1$ is at least $D^2$ times the sum of the next $2L$ largest coordinates of $\bx_{\spa}$ and at least $D^4$ times the sum of the next $2L$ largest coordinates after that, etc. Thus, $\|\bx'\|_1 \ge \frac{1}{2}\|\bx_{\spa}\|$.
   
   Now with probability $1/2^{2L}$, because $\bv$ has i.i.d.~Rachemacher entries, $\langle \bv, \bx'\rangle = \|\bx'\|_1$, and with probability at least $1/2$, $\langle \bv, \bx_{\spa} - \bx'\rangle \ge 0$. Thus, with probability at least $1/2^{2L+1} \ge 2n^{-0.001}$, $\langle \bv, \bx_{\spa}\rangle \ge \frac{1}{2}\|\bx_{\spa}\|_1$. 

  \subsubsection*{Part 2, $\langle \bv, \bx_{\den}\rangle = \Omega(\sqrt{\log n}\|\bx_{\den}\|_2)$} 

  Since $\bx_{\den}  = \sum_{i \in F}\bx^{(i)}$, we have that
  \begin{align}
\Pr\left[\langle \bv, \bx_{\den}\rangle \ge \frac{\sqrt{\log n}}{1000D} \|\bx_{\den}\|_2\right] \ge \prod_{i \in F}\Pr\left[\langle \bv, \bx^{(i)} \rangle \ge \frac{\sqrt{\log n} \cdot \|\bx^{(i)}\|_2^2}{1000D\|\bx_{\den}\|_2}\right]\label{eq:done}
  \end{align}
  
  Consider $i \in F$, and let $m_i \ge L + 1$ be the support size of $\bx^{(i)}$. Since $\bx^{(i)}_j D^i \ge 1$ for all $j$ in the support of $\bx^{(i)}$, we have by Theorem~\ref{thm:erdos-lo} and Claim~\ref{claim:binom-sum}, that for any integer $k \in [0, m_i]$ 
  \[
    \Pr\left[\langle \bv, \bx^{(i)} \rangle \ge \frac{k}{D^i}\right] \ge \sum_{i = \frac{m_i+k}{2}+1}^{m_i} \binom{m_i}{i} \ge \frac{1}{10000}\exp\left(-10m_i\left(\frac{k+2}{m_i}\right)^2\right).
  \]
  Observe that $\|\bx^{(i)}\|_2 \le \sqrt{m_i}\|\bx^{(i)}\|_{\infty} \le \sqrt{m_i} D^{-(i-1)}$. Thus,
\[
\Pr\left[\langle \bv, \bx^{(i)} \rangle \ge \frac{k}{D\sqrt{m_i}}\|\bx^{(i)}\|_2\right] \ge \frac{1}{10000}\exp\left(-10m_i\left(\frac{k+2}{m_i}\right)^2\right).
\]

Let \[
  k = \left\lceil \frac{1}{1000}\sqrt{m_i \log n} \cdot \frac{\|\bx^{(i)}\|_2}{\|\bx_{\den}\|_2}\right\rceil.
\]
Then, note that
\[
\frac{k+2}{m_i} \le \frac{3}{m_i} + \frac{1}{1000}\sqrt{\frac{\log n}{m_i}} \cdot \frac{\|\bx^{(i)}\|_2}{\|\bx_{\den}\|_2}
\]
Thus, since $(a+b)^2 \le 2a^2 + 2b^2$,
\[
-10m_i\cdot\left(\frac{k+2}{m_i}\right)^2 \ge -\frac{180}{m_i} - \frac{\log n}{5\cdot 10^4}\cdot \frac{\|\bx^{(i)}\|_2^2}{\|\bx_{\den}\|_2^2}.
\]
Therefore,
\[
\Pr\left[\langle \bv, \bx^{(i)} \rangle \ge \frac{\sqrt{\log n}}{1000D} \cdot \frac{\|\bx^{(i)}\|_2^2}{\|\bx_{\den}\|_2}\right] \ge \frac{1}{10^4}\exp\left(-\frac{180}{m_i}-\frac{\log n}{5\cdot 10^4}\cdot \frac{\|\bx^{(i)}\|_2^2}{\|\bx_{\den}\|_2^2}\right).
\]
Applying Eq.~\ref{eq:done}, and noting that each $m_i \ge L \ge |F|$. 

\begin{align*}
  \Pr\left[\langle \bv, \bx_{\den}\rangle \ge \frac{\sqrt{\log n}}{1000D} \|\bx_{\den}\|_2\right] &\ge  \frac{1}{10^{4L}}\exp\left(-\sum_{i\in F}\frac{180}{m_i} - \frac{\log n}{5\cdot 10^4}\right)\\
                                                                                                  &= \frac{1}{10^{4L}e^{180}}n^{-10^{-5}}\\
  &\ge n^{-2^{-90}} n^{-10^{-5}}\\
  &\ge n^{-0.001},
\end{align*}
For $n$ sufficiently large. This concludes the proof.
\end{proof}

\end{document}